\documentclass{ifacconf}

\usepackage{graphicx}      
\usepackage{natbib}        
\usepackage{amsmath, amssymb,amsthm}
\usepackage{mathtools}
\usepackage{xcolor}
\usepackage{comment}
\usepackage[shortlabels]{enumitem}   

\newtheorem{theorem}{Theorem}[section]
\newtheorem{lemma}[theorem]{Lemma}

\newtheorem{remark}{Remark}[section]

\newtheorem{definition}[theorem]{Definition}

\newtheorem{example}[theorem]{Example}
\newtheorem{assumption}[theorem]{Assumption}

\DeclareMathOperator*{\argmin}{arg\,min}

\begin{document}
\begin{frontmatter}

\title{Implications of Regret on  Stability of\\ Linear Dynamical Systems\thanksref{footnoteinfo}}

\thanks[footnoteinfo]{This work has been supported by the Swiss National Science Foundation under NCCR Automation (grant agreement $51\text{NF}40\_180545$),  and by the  European Research Council under the ERC Advanced grant agreement  $787845$ (OCAL).}

\author[First]{Aren Karapetyan}
\author[First]{Anastasios Tsiamis} 
\author[First]{Efe C. Balta}
\author[Second]{Andrea Iannelli}
\author[First]{John Lygeros}

\address[First]{Automatic Control Laboratory, ETH Zürich, 
   Zürich 8092, Switzerland (e-mails: \{akarapetyan, atsiamis, ebalta, jlygeros\}@control.ee.ethz.ch)}
\address[Second]{University of Stuttgart, Institute for Systems Theory and Automatic Control, Stuttgart 70569,  Germany\\(e-mail: andrea.iannelli@ist.uni-stuttgart.de)}

\begin{abstract}                
The setting of an agent making decisions under uncertainty and under dynamic constraints is common for the fields of optimal control, reinforcement learning, and recently also for online learning. In the online learning setting, the quality of an agent's decision is often quantified by the concept of regret, comparing the performance of the chosen decisions to the best possible ones in hindsight. While regret is a useful performance measure, when dynamical systems are concerned, it is important to also assess the stability of the closed-loop system for a chosen policy. In this work, we show that for linear state feedback policies and linear systems subject to adversarial disturbances, linear regret implies  asymptotic stability in both time-varying and time-invariant settings. Conversely, we also show that bounded input bounded state stability and summability of the state transition matrices imply linear regret.
\end{abstract}


\end{frontmatter}

\section{Introduction}

A number of real-world problems can be cast into the framework of agents making optimal decisions under uncertainty and/or adversarial disturbances. In this setting, the agent has an associated dynamical system and at each timestep suffers an \emph{a priori} unknown cost that depends on its state and input. In the case of perfect knowledge of the dynamics and the future costs, this can be turned into an optimal control problem and solved with one of the plethora of available methods. As is often the case, however, the dynamics, disturbances, and/or future costs are either entirely unknown or only partially known. This is the setting, for example, in reinforcement learning and approximate dynamic programming \citep{bertsekas2015dynamic}. 

As agents make decisions ``on-the-go'', there is a need to quantify and compare the performance of various algorithms. Considering a closed-loop system with a given, possibly time-varying policy, one can study its asymptotic stability as an asymptotic metric. Another approach is to look into the problem through the lens of online optimization. An important metric in the literature of the latter is the notion of regret. Given a policy $\mu$, its regret $\mathcal{R}_T$, is defined  as the difference between its accumulated cost over some time horizon $T$ and that of some benchmark policy $\pi$. It is often desirable to have sublinear growth of $\mathcal{R}_T$ with respect to $T$, which will achieve average convergence to the benchmark in the limit.

A special case of the optimal control problem, the linear quadratic regulator (LQR) has been extensively studied in this context. The results in \citep{simchowitz2020naive} show that the certainty equivalence approach can synthesize stable linear state feedback controllers as long as the model estimate errors are small enough. In the same spirit of certainty equivalence, the algorithms proposed in \citep{cohen2018online, jedra2022minimal} yield sublinear regret bounds by sequentially solving Riccati equations. The LQR problem with adversarial disturbances has been studied in \citep{yu2020power, zhang2021regret} in the presence of a prediction window, in \citep{karapetyan2022regret} to quantify the regret of robustness in the sense of $\mathcal{H}_{\infty}$ control,  and in a regret-optimal setting in \citep{sabag2021regret} for unconstrained and in \citep{didier2022system, martin2022safe} for constrained cases. The setting with general convex and differentiable cost functions with adversarial disturbances has been studied for linear time-invariant (LTI) \citep{agarwal2019logarithmic} and for linear time-varying (LTV) \citep{gradu2020adaptive} systems; algorithms that achieve sublinear regret bounds have been proposed for both classes of systems.

To study and analyze algorithms in both the control theoretic and online learning contexts, one needs to understand the relationship between regret and stability. While most online learning-inspired works seek sublinear regret in pursuit of suboptimality guarantees, it is unclear whether good performance in this aspect also implies stability. Conversely, it is not known what the absence of such guarantees means even for linear dynamical systems. In a recent work \citep{nonhoff2022relation}, the authors study the interconnection between bounded regret and closed-loop stability for non-linear systems in the absence of noise. In the current work, we consider linear systems subject to process noise and study the relationship between linear regret and stability. Our goal is to develop sufficient conditions under which regret guarantees of an algorithm imply stability of the closed-loop system and vice versa.

In particular, we consider generic, time-varying costs, LTV systems subject to adversarial disturbances with LTV state feedback policies, and also LTI systems as a special case. We study the connection between the regret of such  policies and their closed-loop stability in the sense of bounded input-bounded state (BIBS) stability with respect to disturbances and  asymptotic stability in the absence of disturbances. Under suitable assumptions, we show the following:

\begin{enumerate}[(a)]
    \item For LTV systems and LTV state feedback policies linear regret implies asymptotic stability of the closed-loop system. Conversely, BIBS stability and absolute summability of the state transition matrices imply linear regret.
    \label{question:a}
    \item For LTI systems and LTI state feedback policies linear regret is attained if and only if the closed-loop system is asymptotically stable.
    \label{question:b}
\end{enumerate}

An essential requirement for the results to hold is the  ``observability" of the state with respect to costs. We provide a counterexample of a notable case where the asymptotic stability implication of linear regret fails to hold when this assumption is violated. A numerical example to showcase the result for the LTI case is provided.

 \textit{Notation}: For a square matrix $A$ the spectral radius and the spectral norm are denoted by $\rho(A)$, and  $\|A\|$, respectively. For a vector $w \in \mathbb{R}^{n}$, $\|w\|$ denotes its Euclidean norm. $\mathbb{R}_+$ is the set of positive real numbers and $\mathbb{N}$ that of non-negative integers. 

\section{Problem Setup}

\subsection{Preliminaries}
\label{sec:preliminaries}
Consider discrete-time LTV systems of the form
\begin{equation}
    x_{t+1} = A_tx_t +B_tu_t +w_t,
    \label{eq:linsys_tv}
\end{equation}
evolving over non-negative times $t\in\mathbb{N}$, where  $A_t \in \mathbb{R}^{n \times n}$, $B_t \in \mathbb{R}^{n \times m}$ are real matrices, and $x_t, w_t \in \mathbb{R}^n$, $u_t \in \mathbb{R}^m$ are, respectively, the state, disturbance and input signals at timestep $t$. The disturbances are treated as adversarial and are assumed to be norm bounded, that is, there exists  $W \in \mathbb{R}_+$ such that $\|w_t\|\leq W$ for all $t$. The objective of the optimal control problem is to minimize the cost function
\begin{equation}
    J_T(x_0,\boldsymbol{u};\boldsymbol{w}) = \sum_{t=0}^{T}c_t(x_t,u_t),
    \label{eq:cost_function}
\end{equation}
where $c_t(x_t,u_t):\mathbb{R}^n \times \mathbb{R}^m \rightarrow \mathbb{R}$ is the stage cost function at $t$ and $\boldsymbol{u} := [u_0^{\top} \hdots u_{T-1}^{\top}]^{\top}$, $\boldsymbol{w} := [w_0^{\top} \hdots w_{T-1}^{\top}]^{\top}$. 

\textbf{Benchmark Controller:} If the system matrices $(A_t,B_t)$ in \eqref{eq:linsys_tv}, the stage costs $\{c_t\}_{t=0}^T$  and the disturbance signal $\boldsymbol{w}$ are known at time $0$, then the control input signal $\boldsymbol{u^{\star}}$ that minimises \eqref{eq:cost_function} can be obtained by solving the following optimization problem
\begin{align*}
    &\argmin_{\boldsymbol{u}} \quad J_T(x_0,\boldsymbol{u};\boldsymbol{w})\\
     &\text{subject to} \quad \eqref{eq:linsys_tv} \quad \forall~ 0\leq t<T.
\end{align*}
Note that $\boldsymbol{u^{\star}}$ is referred to as the optimal input signal in hindsight, obtained with full knowledge of the uncertainty.

\textbf{Linear Policies}: In this work, we consider linear time-varying policies of the form\footnote{When the context is clear, we drop the subscript $t$ and/or the argument when referring to the linear policy $\mu_t(x_t)$.} $\mu_t(x_t) = -K_t x_t$ for some $K_t \in \mathbb{R}^{m\times n}$ for all $t\in \mathbb{N}$, where the matrices $K_t$ are independent of the state and input history. This covers the class of offline linear policies discussed in \citep{goel2022power}. Moreover, we note that the considered policies also cover the case where the time-varying parameters of the stage costs are revealed sequentially, thus generating $K_t$ online.  For example, consider the LQR problem with time-varying cost matrices $Q_t$, $R_t$, that depend only on time $t$ and are revealed in the form of predictions, studied in \citep{zhang2021regret}. The  future cost functions, disturbances, as well as the system matrices may be unknown to the policy.

For a given noise signal realisation $\boldsymbol{w}$ and cost functions $\{c_t\}_{t=0}^T$, we define the regret of a given policy $\mu$ to be
\begin{equation*}
    \mathcal{R}_T(\mu;\boldsymbol{w}):= J_T(x_0,\boldsymbol{u}^{\mu};\boldsymbol{w}) - J_T(x_0,\boldsymbol{u^{\star}};\boldsymbol{w}),
    \label{eq:regret_definition_general}
\end{equation*}
where $\boldsymbol{u}^{\mu}\in\mathbb{R}^{Tm}$ is the input signal generated by $\mu$.

Whereas the regret provides intuition on the performance of the given policy with respect to the best possible cost, its implication for the stability of the associated closed-loop system is unclear. Though several works, e.g. in the stochastic setting \citep{dean2018regret, cohen2018online}, also provide stability guarantees for the  online policies, there is no clear relationship between the order of regret and the notion of stability; the results below attempt to address this shortcoming. 

We start by defining linear regret to show the conditions under which it implies stability and vice versa.
\begin{definition}
    A policy $\mu$ is said to have linear regret if  for  given $W,X\in \mathbb{R}_+$, and class of stage costs \footnote{For example, the one defined later in Assumption \ref{assum:cost_functions}.}, there exist ${C}_w,{C}_0 \in \mathbb{R}_+$ such that for all $x_0 \in \mathbb{R}^n$ with $||x_0||\leq X$, 
    and for all admissible sequences $\{c_t\}_{t=0}^T$ and $\{w_t\}_{t=0}^{T-1}$ 
    \begin{equation}
        \mathcal{R}_T(\mu;\boldsymbol{w})\leq C_0 +C_wT, \quad \forall T>0.
        \label{eq:linear_regret}
    \end{equation}
    \label{def:linear_regret}
\end{definition}
In the rest of this section, we review some notions of stability and provide the assumptions required to show the points \ref{question:a} and \ref{question:b}. 

The LTV system in \eqref{eq:linsys_tv}  takes the following closed-loop form for considered policies $\mu_t(x_t) = -K_tx_t$
\begin{equation}
    x_{t+1} = F_tx_t + w_t,
    \label{eq:closed_loop_system}
\end{equation}
where $F_t = A_t - B_tK_t$. We refer to \eqref{eq:closed_loop_system} as the closed-loop system for the considered policy; for questions of asymptotic or exponential stability, we consider the  unforced version of \eqref{eq:closed_loop_system} with $w_t = 0$ for all $t$ \citep{ludyk2013stability}.

Let $\Phi(t,t_0) \in \mathbb{R}^{n \times n}$ denote the state transition matrix of the free system $x_{t+1} = F_tx_t$ at time $t$, starting at some  initial time $t_0$. In the following, we provide the definitions of asymptotic, exponential, and BIBS stability, specialized for LTV systems \citep{callier2012linear, ludyk2013stability}.

\begin{definition} 
(Asymptotic Stability). The LTV system \eqref{eq:closed_loop_system} is  asymptotically stable if and only if for all $x_0 \in \mathbb{R}^n$ and for all $t_0 \in \mathbb{N}$ 
\begin{enumerate}
\item There exists $b\in \mathbb{R}_+$, such that $\|\Phi(t,t_0)x_0\| \leq b \quad$ for all  $t \geq t_0$,
\item $\lim_{t\rightarrow\infty}\Phi(t,t_0)x_0 = \boldsymbol{0}$.\\
\end{enumerate}
\label{def:asymptotic_stability}
\end{definition}

\begin{definition}
(BIBS stability). The LTV system \eqref{eq:closed_loop_system} is BIBS stable if and only if
\begin{equation*}
    \sup_{t>0}\left[\sum_{k=1}^{t}\|\Phi(t,k)\|\right]< \infty.
\end{equation*}
\end{definition}
\begin{definition} (Exponential Stability). The LTV system \eqref{eq:closed_loop_system} is exponentially stable if and only if there exist  $\delta \in [0,1)$ and $d\in \mathbb{R}_+$, such that for all $t_0\in\mathbb{N}$ 
\begin{equation*}
    \|\Phi(t,t_0)\| \leq d \delta^{t-t_0} \quad \forall t\geq t_0.
\end{equation*}
 \label{def:exponential_stability}   
\end{definition}

We make the following assumptions on the dynamics.
\begin{assumption}
(System Dynamics).
\begin{enumerate}[label=\roman*.]
    \item There exists a policy  $\overline{\mu}_t(x_t) = -\overline{K}_tx_t$, such that the closed-loop system \eqref{eq:closed_loop_system} with $F_t = \overline{F}_t := A_t-B_t\overline{K}_t$ is exponentially stable. \label{assum:stabilizability}
    \item The state transition matrix for \eqref{eq:closed_loop_system} satisfies either of  the following, \label{assum:non_chaotic_systems}
\begin{enumerate}
    \item $\lim\limits_{T \rightarrow \infty}{\|\Phi(T,0)\|} \in\mathbb{R}\cup\{\infty\}$\\
    \item $0<\liminf\limits_{T \rightarrow \infty}{\|\Phi(T,0)\|} <\limsup\limits_{T \rightarrow \infty}{\|\Phi(T,0)\|}<\infty$.
\end{enumerate} 
\item The considered linear state feedback policies are such that the system matrix $F_t$ is full rank for all $t\in\mathbb{N}$. \label{assum:full_rank_system_matrix}
\end{enumerate}
\label{assum:system_dynamics}
\end{assumption}
We note that Assumption \ref{assum:system_dynamics}.\ref{assum:stabilizability} is required for performance guarantees of the benchmark policy and corresponds to the stabilizability in the LTI case. Assumption \ref{assum:system_dynamics}.\ref{assum:non_chaotic_systems} excludes chaotic LTV systems in the sense of Li-Yorke \citep{shi2009chaos} and is always satisfied for LTI systems. It can be relaxed, but we introduce it here to keep the discussion simple. Assumption \ref{assum:system_dynamics}.\ref{assum:full_rank_system_matrix} is standard in the discrete-time case to avoid deadbeat-type responses that imply a lack of uniqueness of the solutions backward in time.

Assumption \ref{assum:system_dynamics}.\ref{assum:full_rank_system_matrix} leads to the following lemma \citep{callier2012linear}.
\begin{lemma}
Assume that Assumption \ref{assum:system_dynamics}.\ref{assum:full_rank_system_matrix} holds, then the LTV system \eqref{eq:closed_loop_system} is  asymptotically stable if and only if
\begin{equation}
    \lim_{T \rightarrow \infty}\|\Phi(T,0)\|=0.
    \label{eq:transition_matrix_limit}
\end{equation}
\label{lem:asymptotic_stability}
\end{lemma}
We also introduce the following assumption for the costs.

\begin{assumption}
\label{assum:cost_functions}
(Stage Costs). There exist positive $\overline{M},\underline{M}\in \mathbb{R}_+$ and $\overline{s},\underline{s} \geq 1$, such that\footnote{We consider throughout $\underline{s} = \overline{s} = 2$ without loss of generality.} for all $  x \in\mathbb{R}^n, u\in \mathbb{R}^m$  and $t\in\mathbb{N}$,\\
$$\underline{M}\|x\|^{\underline{s}}\leq c_t(x,u) \leq \overline{M}(\|x\|^{\overline{s}}+\|u\|^{\overline{s}}).$$
\end{assumption}
 
 The lower bound in the above assumption, which we refer to as the ``observability" of the state with respect to cost, excludes cases when unstable states can be ``hidden" in the cost. This condition is common in the literature \citep{bertsekas2015value} and is also referred to as positive definiteness or detectability of the system with respect to stage costs \citep{postoyan2016stability}. 


\subsection{Relation to Tracking Problems}
For completeness, we also point out a connection between the stability discussion below and tracking problems. Consider the problem of tracking an unknown, bounded,  time-varying reference signal $r_t\in \mathbb{R}^n$ given system dynamics \eqref{eq:linsys_tv}, where the agent has access to $r_t$ at timestep $t$, but not before. We can then write the state evolution of the tracking error $\Tilde{x}_t := x_t-r_t$ as 
\begin{equation*}
    \Tilde{x}_{t+1} = A\Tilde{x}_t + Bu_t + \underbrace{w_t - r_{t+1} +Ar_t}_{\nu_t},
\end{equation*}
where $\nu_t \in \mathbb{R}^n$ can be considered as the disturbance of the modified system; note that $r_{t+1}$ and $w_t$ are unknown at time $t$. If the assumptions above are satisfied the tracking problem can be treated as the regulation under the adversarial noise problem in Section \ref{sec:preliminaries}.

\section{Main Results}

In this section, we derive our main results for LTV and LTI systems. We start with the following motivating example of discounted LQR as a particular case when an unstable system attains linear regret if Assumption \ref{assum:cost_functions} is violated. 

\begin{example}
Consider the dynamics \eqref{eq:linsys_tv}, with $A_t =A$, $B_t=B$ for all $t\in \mathbb{N}$, and the following cost function
\begin{equation}
     J_T^d(x_0,\boldsymbol{u};\boldsymbol{w}) = x_T^{\top}P_{\alpha}x_T + \sum_{i=0}^{T-1}\alpha^i\left(x_i^{\top}Qx_i +u_i^{\top}Ru_i\right),
     \label{eq:discounted_cost}
\end{equation}
where $\alpha \in (0,1)$ is a discount factor, $Q \in \mathbb{R}^{n \times n}$, $R \in \mathbb{R}^{m \times m}$ are positive definite matrices, and $P_{\alpha} \in \mathbb{R}^{n \times n}$ satisfies the following equation
\begin{equation*}
    P_{\alpha} = Q + \alpha A^{\top}P_{\alpha}A  - \alpha^2A^{\top}P_{\alpha}B(R+\alpha B^{\top}P_{\alpha}B)^{-1}B^{\top}P_{\alpha}A.
\end{equation*}
The above is the discrete algebraic Riccati equation (DARE) for the modified system $(\sqrt{\alpha}A,B)$ with cost matrices $Q$ and $\frac{R}{\alpha}$. The stage costs do not satisfy the condition in Assumption \ref{assum:cost_functions} as there exists no uniform lower bound. Consider now the certainty equivalent controller that minimises the cost \eqref{eq:discounted_cost} assuming $w_t = \boldsymbol{0} \quad \forall t\in\mathbb{N}$. The optimal policy $\mu_t$ in this setting is then a linear state feedback, given by $\mu_t = -K_{\alpha}x_t$ \citep{bertsekas2015dynamic}, where
\begin{equation*}
    K_{\alpha} =  -\alpha(R+\alpha B^{\top}P_{\alpha}B)^{-1}B^{\top}P_{\alpha}A.
\end{equation*}
We make the hypothesis that the value function (cost-to-go) at time step $t$ associated with this policy is given by 
\begin{equation*}
    V_t(x_t)= \alpha^t\left[x_t^{\top}P_{\alpha}x_t + v_t^{\top}x_t + q_t\right],
\end{equation*}
where $v_t\in\mathbb{R}^n$, $q_t \in \mathbb{R}$. At timestep $T$, it is easily verified with $v_T=\boldsymbol{0}$, $q_T=0$. Assuming the hypothesis is true for $t+1$, the cost-to-go at $t$ is
\begin{align*}
    V_t(x_t) &= \alpha^t[x_t^{\top}(Q+K_{\alpha}^{\top}RK_{\alpha})x_t+ \alpha x_t^{\top}F_\alpha^{\top}P_{\alpha}F_\alpha x_t\\ 
    &+2\alpha x_t^{\top}\left(F_{\alpha}^{\top}P_{\alpha}w_t+\frac{1}{2}v_{t+1}\right)\\
    &+ \alpha\left(w_t^{\top}P_{\alpha}w_t+w_t^{\top}v_{t+1}+q_{t+1}\right)].
\end{align*}
One can verify that the induction hypothesis is verified only if $P_{\alpha}$ satisfies the DARE for the modified system, and
\begin{align*}
    v_t &= 2\alpha F_{\alpha}\left(P_{\alpha}w_t + \frac{1}{2}v_{t+1}\right)\\
    q_t &= \alpha \left( w_t^{\top}P_{\alpha}w_t + w_t^{\top}v_{t+1} + q_{t+1} \right).
\end{align*}
The cost under this linear state feedback policy can then be attained  by applying the preceding recursion repeatedly
\begin{equation*}
J_T^d(x_0,\boldsymbol{u}^{\mu};\boldsymbol{w}) = x_0^{\top}P_{\alpha}x_0 + x_0^{\top}v_0 + q_0,
\end{equation*} 
where
\begin{align*}
    v_0 &= 2 \sum_{k=0}^{T-1}\alpha^{k+1}\left(F_{\alpha}^{\top}\right)^{k+1}P_{\alpha}w_k,\\
    q_0 &= \sum_{k=0}^{T-1}\alpha^{k+1}\left(w_k^{\top}P_{\alpha}w_k + \sum_{i=1}^{T-1}\left(\alpha^i\left(F_{\alpha}^{\top}\right)^i P_{\alpha}w_i\right)\right),
\end{align*}
and $F_{\alpha}:= A-BK_{\alpha}$. It is known that for $\alpha$ small enough the closed-loop system matrix $F_{\alpha}$ may in fact be unstable \citep{postoyan2016stability}. Consider the set,
\begin{equation*}
    \Gamma := \{\alpha \in (0,1)\:|\:\alpha \|F_{\alpha}\|<1, \rho(F_{\alpha})>1\}.
\end{equation*}
Then for some $\bar{\alpha} \in \Gamma$, and for $\sigma:= \bar{\alpha}\|F_{\alpha}\|$
\begin{equation*}
\begin{split}
     J_T^d(x_0,\boldsymbol{u}^{\mu};\boldsymbol{w}) &\leq \|x_0\|^2\|P_{\bar{\alpha}}\| + 2\|x_0\|\|P_{\bar{\alpha}}\|W\frac{\sigma(1-\sigma^T)}{1-\sigma}\\
     &+ W^2\|P_{\bar{\alpha}}\|\frac{\bar{\alpha} (1-\bar{\alpha}^T)}{1-\bar{\alpha}} + 2\|P_{\bar{\alpha}}\|\frac{\sigma}{1-\sigma}WT.
     \end{split}
\end{equation*}
It is then evident that there exist $C_w, C_0 \in \mathbb{R}_+$, such that
\begin{equation*}
J_T^d(x_0,\boldsymbol{u}^{\mu};\boldsymbol{w}) \leq C_0 + C_wT.
\end{equation*}
Moreover, such $C_0$ and $C_w$ exist for all  costs of the form \eqref{eq:discounted_cost}. Hence, for all such $\{c_t\}_{t=0}^T$, the regret will also attain the same bound. Thus, while the considered policy achieves linear regret, the closed-loop with the system is unstable. 
\end{example}
This example shows how the violation of Assumption \ref{assum:cost_functions} can make regret hide   the instability of the system.

\subsection{Linear Time-Varying Systems}

The following theorem establishes the regret-stability relation for the time-varying policy $\mu_t(x_t) = -K_tx_t$.

\begin{theorem}
Assume that Assumption \ref{assum:cost_functions} holds. Given the cost function \eqref{eq:cost_function} and LTV system \eqref{eq:linsys_tv}, consider a linear time-varying state feedback policy $\mu_t(x_t) = -K_tx_t$. If the closed-loop system \eqref{eq:closed_loop_system} under this policy is BIBS stable, and if there exists $D \in \mathbb{R}_+$ such that
\begin{equation}
    \lim_{T \rightarrow \infty}\sum_{t=0}^T\|\Phi(t,0)\|\leq D,
    \label{eq:summable_sum_state_transition}
\end{equation}
then the policy $\mu$ attains linear regret.

\label{the:summable_transiton_linear_regret}
\end{theorem}
\begin{proof}
Using the condition in Assumption \ref{assum:cost_functions} and defining $M: = \overline{M}(1+\mathop{\max}_{0\leq t<T}\|K_t\|^2)$
\begin{align*}
    &J_T(x_0,\boldsymbol{u}^{\mu};\boldsymbol{w}) = \sum_{t=0}^Tc_t(x_t,K_tx_t) \leq M\sum_{t=0}^T\|x_t\|^2\\
    & = M\sum_{t=0}^T\left\|\Phi(t,0)x_0 + \sum_{k=1}^t\Phi(t,k)w_{k-1}\right\|^2\\
    & \leq 2M\|x_0\|^2\sum_{t=0}^T\|\Phi(t,0)\|^2 + 2MW^2\sum_{t=0}^T\sum_{k=1}^t\|\Phi(t,k)\|^2\\
    & \leq 2M\overline{D}X^2 + 2M\overline{H}W^2T,
\end{align*}
where $\overline{D},\overline{H} \in \mathbb{R}_+$ are such that
\begin{align}
\label{eq:summable_state_transition_s}
    \lim_{T \rightarrow \infty}\sum_{t=0}^T\left\|\Phi(t,0)\right\|^2\leq \overline{D} < \infty,\\
    \label{eq:BIBS_upper_bound_s}
    \lim_{t \rightarrow \infty}\sum_{k=1}^t\|\Phi(t,k)\|^2\leq \overline{H} < \infty.
\end{align}
The limit in \eqref{eq:summable_state_transition_s} exists due to the condition in \eqref{eq:summable_sum_state_transition} and the fact that the series contains only non-negative terms. The limit in \eqref{eq:BIBS_upper_bound_s} exists due to the BIBS assumption and the series of non-negative terms only. Since $J_T(x_0,\boldsymbol{u};\boldsymbol{w})\geq 0$ for all $\boldsymbol{u}$, the regret will necessarily attain the same bound with $C_w = 2M\overline{H}W^2$ and $C_0 = 2M\overline{D}X^2$.
\end{proof}

We note that given Assumption \ref{assum:system_dynamics}.\ref{assum:full_rank_system_matrix} holds, the condition in \eqref{eq:summable_sum_state_transition} is  stronger than  asymptotic stability. In fact, if \eqref{eq:summable_sum_state_transition} is satisfied, then Assumption \ref{assum:system_dynamics}.\ref{assum:full_rank_system_matrix} implies asymptotic stability of the closed-loop system.

The following theorem establishes the implication of linear regret on the stability of the closed-loop system with linear time-varying state feedback policies.
\begin{theorem}
Assume that Assumptions \ref{assum:system_dynamics}  and \ref{assum:cost_functions} hold. Given the cost function \eqref{eq:cost_function} and LTV system \eqref{eq:linsys_tv}, if a linear time-varying  state feedback policy $\mu_t(x_t) = -K_tx_t$ attains linear regret, then the closed-loop system \eqref{eq:closed_loop_system} under this policy is asymptotically stable.
\label{the:regret_implying_stability_tv}
\end{theorem}
\begin{proof}
Given Assumption \ref{assum:system_dynamics}.\ref{assum:stabilizability} holds, there exists a policy $\overline{\mu}_t = -\overline{K}_tx_t$ such that the closed-loop system \eqref{eq:closed_loop_system} under this policy is exponentially stable. It can be shown that this implies BIBS stability \citep{ludyk2013stability}, and   the bound in \eqref{eq:summable_sum_state_transition} is always satisfied. Specifically, from the  definition of exponential stability, it can be inferred that taking $t_0 = 0$ and $D =\frac{d}{1-\delta}$ the bound in \eqref{eq:summable_sum_state_transition} is achieved. Therefore, from Theorem \ref{the:summable_transiton_linear_regret} there exist some $C_w^{\star},C_0^{\star}\in \mathbb{R}_+$, such that
\begin{equation*}
     J_T(x_0,\boldsymbol{u^{\star}};\boldsymbol{w})\leq J_T(x_0,\boldsymbol{u}^{\overline{\mu}};\boldsymbol{w})\leq  C_0^{\star} +C_w^{\star}T, 
\end{equation*}
where the first inequality follows from the definition of $\boldsymbol{u^{\star}}$. Then, linear regret of any policy $\mu$ is equivalent to
\begin{equation}
    J_T(x_0,\boldsymbol{u}^{\mu};\boldsymbol{w}) \leq  \overline{C}_0 + \overline{C}_wT,
    \label{eq:regret_bound_time_varying_2}
\end{equation}
where $\overline{C}_w = (C_w^{\star}+C_w)$ and $\overline{C}_0 = (C_0^{\star}+C_0)$. Assume, for the sake of contradiction,  that there exists a policy $\mu_t(x_t) = -K_tx_t$, that attains the bound in \eqref{eq:regret_bound_time_varying_2} but the associated closed-loop system \eqref{eq:closed_loop_system} is not asymptotically stable. 
This is equivalent (Lemma \ref{lem:asymptotic_stability} and Assumption \ref{assum:system_dynamics}.\ref{assum:non_chaotic_systems}) to assuming that either the limit in \eqref{eq:transition_matrix_limit} does not exist or it is greater than zero, or equivalently
\begin{equation}
 \limsup_{T\rightarrow\infty}\|\Phi(T,0)\| \geq \liminf_{T\rightarrow\infty}\|\Phi(T,0)\| > 0.
 \label{eq:non_asymptotic_assumption}
\end{equation}
Using the lower bound in Assumption \ref{assum:cost_functions}, we then have
\begin{align*}
   J_T(x_0,\boldsymbol{u}^{\mu};\boldsymbol{w}) &=  \sum_{t=0}^Tc_t(x_t,K_tx_t) \geq \underline{M}\sum_{t=0}^T\|x_t\|^2\\
   &= \underline{M}\sum_{t=0}^T\|\Phi(t,0)x_0 + \sum_{k=1}^t\Phi(t,k)w_{k-1}\|^2.
\end{align*}
Consider first the case when $\liminf\limits_{T\rightarrow\infty}\|\Phi(T,0)\| = \infty$. Then, necessarily there exists a $\bar{x}_0 \neq \boldsymbol{0}$, such, that $\liminf\limits_{T\rightarrow\infty}\|\Phi(T,0)\bar{x}_0\| = \infty$. For $x_0=\bar{x}_0$ and $w_t=\boldsymbol{0}$ for all $t\in\mathbb{N}$, one can take the time-averaged limit of the above
\begin{align*}
    \liminf_{T\rightarrow \infty}\frac{1}{T}J_T(\bar{x}_0,\boldsymbol{u}^{\mu};\boldsymbol{w})&\geq \liminf_{T\rightarrow \infty} \frac{\underline{M}}{T}\sum_{t=0}^T\|\Phi(t,0)\bar{x}_0\|^2\\
    &\geq \liminf_{T\rightarrow \infty} \underline{M}\|\Phi(T,0)\bar{x}_0\|^2 = \infty,
\end{align*}
where the last inequality follows from the Stolz-C\'esaro theorem \citep{choudary2014real}. This leads to a contradiction with the upper bound in \eqref{eq:regret_bound_time_varying_2}
\begin{equation*}
    \liminf_{T\rightarrow \infty}\frac{1}{T}J_T(\bar{x}_0,\boldsymbol{u}^{\mu};\boldsymbol{w})\leq \overline{C}_w.
\end{equation*}

Given Assumption \ref{assum:system_dynamics}.\ref{assum:non_chaotic_systems}, the only other remaining case is when the upper and lower limits in \eqref{eq:non_asymptotic_assumption} are finite (whether equal or not) and non-zero. For this case, consider $x_0 = \boldsymbol{0}$ and $w_{k-1} = C\Phi(k,0)\overline{w}_0$, where
\begin{equation*}
    C =   
   \min_{k\geq0} \left[ \frac{W}{\|\Phi(k,0)\overline{w}_0\|}\right],
\end{equation*}
and the vector $\overline{w}_0 \in \mathbb{R}^n$ is non-zero. Since $\limsup\limits_{T\rightarrow\infty}\|\Phi(T,0)\|$ is finite, such a $C\in\mathbb{R}_+$ exists. Then
\begin{align*}
     &J_T(x_0,\boldsymbol{u}^{\mu};\boldsymbol{w}) \geq\underline{M}C^2\sum_{t=0}^T\|\sum_{k=1}^t\Phi(t,k)\Phi(k,0)\overline{w}_0\|^2\\
     & = \underline{M}C^2\sum_{t=1}^T\|\sum_{k=1}^t\Phi(t,0)\overline{w}_0\|^2 = \underline{M}C^2\sum_{t=1}^Tt^2\|\Phi(t,0)\overline{w}_0\|^2.
\end{align*}

By taking the time-averaged limits of both sides of the inequality and applying the Stolz-C\'esaro theorem again
\begin{align*}
    \liminf_{T\rightarrow \infty}\frac{1}{T}J_T(\bar{x}_0,\boldsymbol{u}^{\mu};\boldsymbol{w})\geq \underline{M}C^2\liminf_{T\rightarrow\infty}T^2\|\Phi(T,0)\overline{w}_0\|^2.
\end{align*}
The lower bound becomes unbounded, contradicting \eqref{eq:regret_bound_time_varying_2} and completing the proof by contraposition.
\end{proof}
Thus, under suitable conditions, linear regret guarantees asymptotic stability. However, as shown in Theorem \ref{the:summable_transiton_linear_regret} and also pointed out in \citep{nonhoff2022relation} for a similar setting, the converse is not always the case, i.e. asymptotic stability alone  does not imply linear regret.
\begin{remark}
The results can be generalized to affine policies of the form $\mu(x_t) = -K_tx_t + d_t$, for some $d_t \in \mathbb{R}^m$, given that $d_t$ are bounded for all $t\in\mathbb{N}$.
\end{remark}
\begin{remark}
Note that inferring stability from regret requires not only stabilizability (Assumption \ref{assum:system_dynamics}.\ref{assum:stabilizability}) but also certain conditions on the benchmark. Indeed to guarantee (\ref{eq:regret_bound_time_varying_2}), it is necessary that a stabilizing policy exists in the feasible space of the benchmark. This holds here since the benchmark is the non-causal controller $\boldsymbol{u^{\star}}$ (dynamic regret), but weaker benchmarks that are constrained to specific policy classes can also be considered (policy regret). 
\end{remark}

\subsection{Linear Time-Invariant Systems}

In this section, we consider the optimal control problem of minimizing the cost \eqref{eq:cost_function} subject to the LTI dynamics
\begin{equation}
    x_{t+1} = Ax_t +Bu_t +w_t,
    \label{eq:linsys}
\end{equation}
where $A\in \mathbb{R}^{n \times n}$, $B\in \mathbb{R}^{n\times m}$. The following establishes the relation between linear regret and stability in this setting. 

\begin{theorem}
Assume that Assumptions \ref{assum:system_dynamics}.\ref{assum:stabilizability},  and \ref{assum:cost_functions} hold. Given the cost function \eqref{eq:cost_function} and LTI system \eqref{eq:linsys}, a linear stationary state feedback policy $\mu(x_t) = -Kx_t$ attains linear regret if and only if the closed-loop system $x_{t+1} =Fx_t$ under this policy is asymptotically stable.
\label{the:time_invariant_regret}
\end{theorem}

\begin{proof}
To show that the required regret bound is achieved for an asymptotically stabilizing  stationary policy, consider the state at timestep $0\leq t <T$, given by
\begin{equation*}
    x_t = F^tx_0 + \sum_{k=0}^{t-1}F^kw_{t-k-1},
\end{equation*}
taking $w_{-1}=\boldsymbol{0}$. Using the cost bounds in Assumption \ref{assum:cost_functions} and denoting $M:=\overline{M}(1+\|K\|^2)$
\begin{align*}
    &\sum_{t=0}^{T}c_t(x_t,Kx_t) \leq \sum_{t=0}^{T}\overline{M}\left(\|x_t\|^2+\|K\|^2\|x_t\|^2 \right)\\
    &= M \sum_{t=0}^T\|F^tx_0 + \sum_{k=0}^{t-1}F^kw_{t-k-1}\|^2\\
    & \leq 2Mg^{2}\sum_{t=0}^T\left(\varepsilon^{2t}\|x_0\|^2 + W^{2} \left(\frac{1-\varepsilon^{t}}{1-\varepsilon}\right)^2\right)\\
    &\leq\frac{ 2Mg^2}{1-\varepsilon^2}X^2 + \frac{2Mg^2}{(1-\varepsilon)^2}W^2T,
\end{align*}

where we used the fact that for $\rho(F)<1$  there exists a $g\in \mathbb{R}_+$, such that $\|F^k\|\leq g \varepsilon^k$ for all $k>0$, where $\varepsilon:=\frac{1+\rho(F)}{2} \in (0,1)$ \citep{horn2012matrix}. Since costs are non-negative, regret attains the same bound.

To prove the reverse statement, assume, for the sake of contradiction, that there exists a  matrix $K'$ such that $\rho(A-BK')\geq1$ attaining linear regret. From the previous analysis, any stabilizing  state feedback matrix $K$ attains a cost that scales linearly with the time horizon. Moreover, such a matrix exists as for LTI systems Assumption \ref{assum:system_dynamics}.\ref{assum:stabilizability} corresponds to the stabilizability of the pair $(A,B)$. Then, using the same arguments as in the proof of Theorem \ref{the:regret_implying_stability_tv}, there exist  $\overline{C}_w, \overline{C}_0 \in \mathbb{R}_+$  such that
\begin{equation}
    \sum_{t=0}^Tc_t(x_t,K'x_t) \leq  \overline{C}_0 + \overline{C}_wT.
    \label{eq:regret_bound_contradiction}
\end{equation}
From Assumption \ref{assum:cost_functions} it is true that
\begin{equation*}
    \sum_{t=0}^{T} c_t(x_t,K'x_t) \geq   \underline{M}\sum_{t=0}^T\|F'^tx_0 + \sum_{k=0}^{t-1}F'^kw_{t-k-1}\|^2.
\end{equation*}

Since the result should hold for any initial state and any disturbance within the defined Euclidian ball, consider $x_0 = \boldsymbol{0}$ and $w_t = \overline{w} \quad\forall t\in\mathbb{N}$ such that $\|\overline{w}\|= W$ and $F'\overline{w} = \rho(F')\overline{w}$. We then have
\begin{align*}
    & \underline{M}\sum_{t=0}^T\|\sum_{k=0}^{t-1}F'^kw_{t-k-1}\|^{2} =   \underline{M}\sum_{t=0}^T\|\sum_{k=0}^{t-1}F'^k\overline{w}\|^{2} \\
    &= \underline{M}\sum_{t=0}^T\|\sum_{k=0}^{t-1}\rho(F')^k\overline{w}\|^{2} \geq  \underline{M}\sum_{t=0}^Tt^{2}W^{2} \geq  \underline{M}W^{2}\frac{T^2+T}{2},
\end{align*}
leading  to a contradiction with \eqref{eq:regret_bound_contradiction}.
\end{proof}
In contrast to the LTV case, in the LTI setting, asymptotic and exponential stability are equivalent, leading to an equivalency between linear regret and stability.

\section{Numerical Example}

To visualize the necessary and sufficient condition in Theorem \ref{the:time_invariant_regret},  a simple two-dimensional system with single input is considered. In particular, for $A = [ 1\quad 1; 0\quad 1]$ and $B = [1;0.5]$, three LTI state feedback controllers are considered, $K_1 = [ 0.2\quad 0.4]$, $K_2 = [0 \quad 1]$, and $K_3 = [-0.02 \quad 0.5]$. These produce respectively, stable, marginally stable, and unstable closed-loop systems. The cost function is taken to be quadratic with $Q = [1.5 \quad 0; 0 \quad 1.5]$ and $R = 1$ as the state and input weighting matrices, respectively. The disturbance $w_t$ for all $0\leq t<T$ is taken to be the normalized eigenvector of the closed-loop system matrix corresponding to the largest eigenvalue, to approximate the worst-case regret for the given policy. The time-averaged regret for each of the controllers is calculated for a time horizon ranging from $1$ to $100$ and is plotted in a logarithmic scale in Figure \ref{fig:LTI_plots}. Specifically, the average regret  of  the stable controller can be upper bounded by a constant, that of the marginally stable controller scales with $T$ and for the unstable one with  a higher order of $T$.

\begin{figure}
\begin{center}
\includegraphics[width=8.4cm]{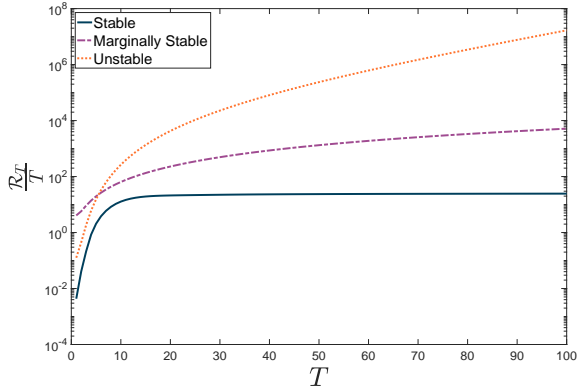}    
\caption{Time-averaged regret, $\frac{\mathcal{R}_T}{T}$ for a LTI system on a semilog scale. The stable system can be upper bounded by a constant while the marginally stable and unstable ones scale with an order of $log(T)$.} 
\label{fig:LTI_plots}
\end{center}
\end{figure}

\section{Conclusions}
In this work, we studied the interconnection of the notion of regret coming from online optimization and the control theoretic concept of stability. Given a linear state feedback policy that attains linear regret, and certain upper and lower bounds on the objective stage costs, we show that the closed-loop system is necessarily  asymptotically stable, both for the time-varying and time-invariant cases. The converse result also holds given that the closed-loop system is BIBS stable and has absolute summable norms of its state transition matrices. The results can be used to directly prove the stability of algorithms with regret guarantees and vice versa.  This work can be a stepping stone for the consideration of adaptive policies, under which the considered setting is no longer linear; this will allow the analysis of a wider range of online algorithms.


\bibliography{ifacconf}             

\begin{thebibliography}{22}
\providecommand{\natexlab}[1]{#1}
\providecommand{\url}[1]{\texttt{#1}}
\providecommand{\urlprefix}{URL }
\expandafter\ifx\csname urlstyle\endcsname\relax
  \providecommand{\doi}[1]{doi:\discretionary{}{}{}#1}\else
  \providecommand{\doi}{doi:\discretionary{}{}{}\begingroup
  \urlstyle{rm}\Url}\fi

\bibitem[{Agarwal et~al.(2019)Agarwal, Hazan, and
  Singh}]{agarwal2019logarithmic}
Agarwal, N., Hazan, E., and Singh, K. (2019).
\newblock Logarithmic regret for online control.
\newblock \emph{Advances in Neural Information Processing Systems}, 32.

\bibitem[{Bertsekas(2015{\natexlab{a}})}]{bertsekas2015dynamic}
Bertsekas, D.P. (2015{\natexlab{a}}).
\newblock Dynamic programming and optimal control 4th edition, volume ii.
\newblock \emph{Athena Scientific}.

\bibitem[{Bertsekas(2015{\natexlab{b}})}]{bertsekas2015value}
Bertsekas, D.P. (2015{\natexlab{b}}).
\newblock Value and policy iterations in optimal control and adaptive dynamic
  programming.
\newblock \emph{IEEE transactions on neural networks and learning systems},
  28(3), 500--509.

\bibitem[{Callier and Desoer(2012)}]{callier2012linear}
Callier, F.M. and Desoer, C.A. (2012).
\newblock \emph{Linear system theory}.
\newblock Springer Science \& Business Media.

\bibitem[{Choudary and Niculescu(2014)}]{choudary2014real}
Choudary, A.D.R. and Niculescu, C.P. (2014).
\newblock \emph{Real analysis on intervals}.
\newblock Springer.

\bibitem[{Cohen et~al.(2018)Cohen, Hasidim, Koren, Lazic, Mansour, and
  Talwar}]{cohen2018online}
Cohen, A., Hasidim, A., Koren, T., Lazic, N., Mansour, Y., and Talwar, K.
  (2018).
\newblock Online linear quadratic control.
\newblock In \emph{International Conference on Machine Learning}, 1029--1038.
  PMLR.

\bibitem[{Dean et~al.(2018)Dean, Mania, Matni, Recht, and Tu}]{dean2018regret}
Dean, S., Mania, H., Matni, N., Recht, B., and Tu, S. (2018).
\newblock Regret bounds for robust adaptive control of the linear quadratic
  regulator.
\newblock \emph{Advances in Neural Information Processing Systems}, 31.

\bibitem[{Didier et~al.(2022)Didier, Sieber, and Zeilinger}]{didier2022system}
Didier, A., Sieber, J., and Zeilinger, M.N. (2022).
\newblock A system level approach to regret optimal control.
\newblock \emph{IEEE Control Systems Letters}.

\bibitem[{Goel and Hassibi(2022)}]{goel2022power}
Goel, G. and Hassibi, B. (2022).
\newblock The power of linear controllers in lqr control.
\newblock In \emph{2022 IEEE 61st Conference on Decision and Control (CDC)},
  6652--6657. IEEE.

\bibitem[{Gradu et~al.(2020)Gradu, Hazan, and Minasyan}]{gradu2020adaptive}
Gradu, P., Hazan, E., and Minasyan, E. (2020).
\newblock Adaptive regret for control of time-varying dynamics.
\newblock \emph{arXiv preprint arXiv:2007.04393}.

\bibitem[{Horn and Johnson(2012)}]{horn2012matrix}
Horn, R.A. and Johnson, C.R. (2012).
\newblock \emph{Matrix analysis}.
\newblock Cambridge university press.

\bibitem[{Jedra and Proutiere(2022)}]{jedra2022minimal}
Jedra, Y. and Proutiere, A. (2022).
\newblock Minimal expected regret in linear quadratic control.
\newblock In \emph{International Conference on Artificial Intelligence and
  Statistics}, 10234--10321. PMLR.

\bibitem[{Karapetyan et~al.(2022)Karapetyan, Iannelli, and
  Lygeros}]{karapetyan2022regret}
Karapetyan, A., Iannelli, A., and Lygeros, J. (2022).
\newblock On the regret of $\mathcal{H}_{\infty}$ control.
\newblock In \emph{2022 IEEE 61st Conference on Decision and Control (CDC)},
  6181--6186. IEEE.

\bibitem[{Ludyk(2013)}]{ludyk2013stability}
Ludyk, G. (2013).
\newblock \emph{Stability of time-variant discrete-time systems}, volume~5.
\newblock Springer-Verlag.

\bibitem[{Martin et~al.(2022)Martin, Furieri, D{\"o}rfler, Lygeros, and
  Ferrari-Trecate}]{martin2022safe}
Martin, A., Furieri, L., D{\"o}rfler, F., Lygeros, J., and Ferrari-Trecate, G.
  (2022).
\newblock Safe control with minimal regret.
\newblock In \emph{Learning for Dynamics and Control Conference}, 726--738.
  PMLR.

\bibitem[{Nonhoff and M{\"u}ller(2022)}]{nonhoff2022relation}
Nonhoff, M. and M{\"u}ller, M.A. (2022).
\newblock On the relation between dynamic regret and closed-loop stability.
\newblock \emph{arXiv preprint arXiv:2209.05964}.

\bibitem[{Postoyan et~al.(2016)Postoyan, Bu{\c{s}}oniu, Ne{\v{s}}i{\'c}, and
  Daafouz}]{postoyan2016stability}
Postoyan, R., Bu{\c{s}}oniu, L., Ne{\v{s}}i{\'c}, D., and Daafouz, J. (2016).
\newblock Stability analysis of discrete-time infinite-horizon optimal control
  with discounted cost.
\newblock \emph{IEEE Transactions on Automatic Control}, 62(6), 2736--2749.

\bibitem[{Sabag et~al.(2021)Sabag, Goel, Lale, and Hassibi}]{sabag2021regret}
Sabag, O., Goel, G., Lale, S., and Hassibi, B. (2021).
\newblock Regret-optimal full-information control.
\newblock \emph{arXiv preprint arXiv:2105.01244}.

\bibitem[{Shi and Chen(2009)}]{shi2009chaos}
Shi, Y. and Chen, G. (2009).
\newblock Chaos of time-varying discrete dynamical systems.
\newblock \emph{Journal of Difference Equations and applications}, 15(5),
  429--449.

\bibitem[{Simchowitz and Foster(2020)}]{simchowitz2020naive}
Simchowitz, M. and Foster, D. (2020).
\newblock Naive exploration is optimal for online {LQR}.
\newblock In \emph{International Conference on Machine Learning}, 8937--8948.
  PMLR.

\bibitem[{Yu et~al.(2020)Yu, Shi, Chung, Yue, and Wierman}]{yu2020power}
Yu, C., Shi, G., Chung, S.J., Yue, Y., and Wierman, A. (2020).
\newblock The power of predictions in online control.
\newblock \emph{Advances in Neural Information Processing Systems}, 33,
  1994--2004.

\bibitem[{Zhang et~al.(2021)Zhang, Li, and Li}]{zhang2021regret}
Zhang, R., Li, Y., and Li, N. (2021).
\newblock On the regret analysis of online {LQR} control with predictions.
\newblock In \emph{2021 American Control Conference (ACC)}, 697--703. IEEE.

\end{thebibliography}
                                                   







\end{document}